\documentclass[letterpaper, 10 pt, conference]{ieeeconf} 

\IEEEoverridecommandlockouts     
\usepackage[hidelinks]{hyperref}
	\hypersetup{
		colorlinks={false},
		citecolor={black},
		urlcolor={black}
	}

\usepackage[dvipsnames]{xcolor}
\usepackage{graphicx}
\usepackage[utf8]{inputenc}
\usepackage{nicefrac}
\usepackage{theorem}
\usepackage{amsmath,amssymb}
\usepackage{enumerate}
\usepackage{placeins}
\usepackage{booktabs}
\usepackage{multirow}
\usepackage{multicol}
\usepackage{cite}
\usepackage{float} 

\usepackage{multirow}

\usepackage{catchfile}

\usepackage{algorithm2e}

\newcommand\complex{{\mathbb C}}
\newcommand\myreal{{\mathbb R}}

\newcommand{\smat}[1]{\ensuremath{\left[ \begin{smallmatrix}#1\end{smallmatrix} \right]}}
\newcommand{\bmat}[1]{\ensuremath{\begin{bmatrix}#1\end{bmatrix}}}

\newcommand{\A}{\ensuremath{\mathcal{A}}}
\newcommand{\Geq}{\ensuremath{\succeq}}
\newcommand{\Gst}{\ensuremath{\succ}}
\newcommand{\Leq}{\ensuremath{\preceq}}
\newcommand{\Lst}{\ensuremath{\prec}}

\newcommand{\He}[1]{\operatorname{He}\left( #1 \right)}

\renewcommand{\Re}{\operatorname{Re}}
\renewcommand{\Im}{\operatorname{Im}}

\DeclareMathOperator{\e}{e}

\newcommand{\param }[1]{\mathbf{\color{Orange}#1}}
\newcommand{\paramA}[1]{\mathbf{\color{ProcessBlue}#1}}
\newcommand{\paramB}[1]{\mathbf{\color{Plum}#1}}

\newtheorem{definition}{Definition}

\newtheorem{proposition}{Proposition}

\newtheorem{remark}{Remark}

\title{\LARGE \bf
Controlling {identical} linear multi-agent systems over directed graphs
}

\author{
   Nicola Zaupa,
   Luca Zaccarian, 
   Sophie Tarbouriech,
   Isabelle Queinnec,
   Giulia Giordano
   \thanks{
      This work was supported in part by ANR through Grant HANDY under Grant ANR-18-CE40-0010; and in part by MUR PRIN Grant DOCEAT under Grant 2020RTWES4.
   }
   \thanks{N. Zaupa, I. Queinnec and S. Tarbouriech are with LAAS-CNRS, University of Toulouse, CNRS, 31400 Toulouse, France. Email: \texttt{\{nzaupa,queinnec,tarbour\}@laas.fr}}
   \thanks{G. Giordano is with Department of Industrial Engineering, University of Trento, 38123 Trento, Italy. Email: \texttt{giulia.giordano@unitn.it}}
   \thanks{L. Zaccarian is with LAAS-CNRS, University of Toulouse, CNRS, 31400 Toulouse, France and with Department of Industrial Engineering, University of Trento, 38123 Trento, Italy. Email: \texttt{zaccarian@laas.fr}}
}

\renewcommand\baselinestretch{0.99}

\begin{document}

\maketitle

\pagestyle{empty}

\begin{abstract}
   We consider the problem of synchronizing a multi-agent system (MAS) composed of several identical
   linear systems connected through a directed graph.
   To design a suitable controller, we construct conditions based on
   Bilinear Matrix Inequalities (BMIs) that ensure state synchronization.
   Since these conditions are non-convex, we propose an iterative algorithm based on a suitable relaxation that allows us to formulate Linear Matrix Inequality (LMI) conditions.
   As a result, the algorithm yields a common static state-feedback matrix for the controller that satisfies general 
   linear performance constraints.
   Our results are achieved under the mild assumption that the graph is time-invariant and connected.
\end{abstract}

\section{Introduction}
   In the last decades, the study of networks, and in particular the distributed control of networked MAS, has attracted a lot of interest in systems and control, due to the broad range of applications in many different areas \cite{olfati-saber2007}, including: power systems, biological systems, sensors network, cooperative control of unmanned aerial vehicles, quality-fair delivery of media content, formation control of mobile robots, and synchronization of oscillators.
   In networked MAS, the general objective is to reach an agreement on a variable of interest.

   We focus our attention on the synchronization problem, where the goal is {to reach} a common trajectory for all agents. 
   In the literature, we can find several studies on scalar agents, but recent works also address networks of agents with finite-dimensional linear input-output dynamics \cite{scardovi2009}.
   In the case of identical agents,
   a common static control law inducing state synchronization
   can be designed by exploiting the information exchange among the agents, which modifies the system dynamics. 
   This exchange is modeled by a (directed or undirected) graph. 
   The spectrum of the Laplacian matrix of this graph plays an important role 
   in the evolution of the associated networked system \cite{MesbahiBook2010}.

   Necessary and sufficient conditions ensuring synchronization have been given under several different forms depending on the context
   \cite{xiao2007,scardovi2009,ma2010,li2010,dalcol2017}.
   A set of necessary and sufficient conditions for identical SISO agents over arbitrary time-invariant graphs is summarized in \cite{interconnected2023}.

   Different approaches
   for control design can be found
   in the literature depending on the desired objective.
   Most of the results are based on the solution of an algebraic Riccati equation, under the assumption that the static control law has a given {infinite gain margin} structure \cite{qin2015,IsidoriBook2017_ch5,BulloBook2022,SaberiBook2022}: the {state-feedback matrix} $K$ has the form $K=B^\top P$, where $B$ is the input matrix and $P$ is the solution of an algebraic Riccati equation. 
   {However,} {imposing} {an infinite gain margin potentially limits} the {achievable} performance.
   As shown in \cite{qin2015}, by choosing a small enough {constant}, {a feedback law can be designed} 
   without knowing the network topology;
   in practice, this constant {depends on the non-zero} eigenvalue of the Laplacian {matrix} having the smallest real part.
   While \cite{IsidoriBook2017_ch5} studies the output feedback case, we consider the dual case, which is also discussed in \cite{BulloBook2022}.
   The design procedure in \cite{trentelman2013} allows achieving  synchronization under bounded $H_\infty$ disturbances, thanks to an observer-based dynamic controller, expressed in terms of suitable algebraic Riccati equations, which guarantee{s} disturbance rejection {properties}.
   {A different} approach base{d} on LMIs is presented in \cite{listmann2015}, where synchronization conditions are imposed {by} {relying on} strong assumptions on the structure of the Lyapunov matrices, while the 
   {problem size is independent of}
   the number of agents. 

   In this work, we study the design of a static state-feedback control law ensuring {MAS} synchronization. 
   The agents are modeled {as} identical LTI {subsystems} and their interconnections are described by time-invariant, directed, and connected graphs.
   We introduce a design strategy based on LMIs, similar to the one in \cite{listmann2015}, 
   but without imposing any assumption on the controller structure or constraints on the Lyapunov matrices, thus ensuring higher degrees of freedom in the design, and potentially improved optimized stabilizers. 
   Through a relaxation of the conditions in \cite{interconnected2023}, we formulate an iterative LMI-based procedure 
   to design a static state-feedback control law.
   Our LMI formulation allows us to easily embed additional
   linear constraints in order to reach a desired performance \cite{BoydBookLMI1994}.

\textbf{Notation.} 
$\myreal$ and $\mathbb{C}$ {denote} the sets of real and complex numbers, respectively. 
We denote with $\jmath$ the imaginary unit.
Given $\lambda = a+\jmath b\in\mathbb{C}$, $\Re(\lambda)=a$ and $\Im(\lambda) = b$ are its real and imaginary parts, respectively;
$\lambda^* = a-\jmath b$ is its complex conjugate. 
$I_N$ is the identity matrix of size $N$,
{while $\mathbf{1}_N \in \mathbb{R}^N$ denotes the $N$ dimensional (column) vector {with} all $1$ entries.}
For any matrix $A$, $A^\top$ denotes the transpose of $A$.
Given two matrices $A$ and $B$, $A\otimes B$ indicates their Kronecker product.
Given a complex matrix ${A} \in \complex^{n\times m}$, ${A}^*$ denotes its conjugate transpose and $\He{A}=A+A^*$. 
Matrix ${A} \in \complex^{n \times n}$ 
is Hermitian if ${A} = {A}^*$, namely $\Re(A)$ is symmetric {$\left(\Re(A)=\Re(A)^\top\right)$} and $\Im(A)$ is skew-symmetric {$\left(\Im(A)=-\Im(A)^\top\right)$}. 
We denote the Euclidean distance of a point $x$ from a set ${\mathcal A}$ as $|x|_{\mathcal A}$.

\section{Problem Statement}
\label{sec:statement}
   Consider $N$ identical dynamical systems
   \begin{subequations} \label{eq:interconnection}
      \begin{equation}\label{eq:sys}
         \dot {x}_i  = A x_i + B u_i
         \quad\quad i=1, \dots, N,
      \end{equation}
      with state vector $x_i \in \mathbb{R}^n$, input vector $u_i \in \mathbb{R}^{m}$, state matrix $A\in\myreal^{n\times n}$ and input matrix $B\in\myreal^{n\times m}$. 
      {Assume} that the pair $(A,B)$ is controllable.
      The directed graph $\mathcal G$ with  {weight} { matrix $\mathcal W\in \mathbb{R}^{N \times N}$} captures the communication topology
       among the agents;
      its Laplacian matrix is $L:={\operatorname{diag}(\mathcal W \mathbf{1}_N)}-\mathcal W$.
      Denote by $0=\lambda_0, \lambda_1, \dots,\lambda_\nu$ the eigenvalues of $L$, ordered with non-decreasing real part 
      (complex conjugate pairs and repeated eigenvalues are only counted once).
      The control input $u_i$ {affecting} agent $i$ is expressed as
      \begin{align}\label{eq:u}
         u_i =  K\sum_{j=1}^{N}{\mathcal W}_{ij}(x_j-x_i)= -K\sum_{j=1}^{N}{L}_{ij}x_j,
      \end{align}
   \end{subequations}
   {where ${\mathcal W}_{ij}$ and $L_{ij}$ are the entries of the {weight} and Laplacian matrices, respectively,}
   and $K\in\myreal^{m\times n}$ is the {state-feedback matrix}. 
   {Each} agent uses only relative information with respect to {the} others, as typically desired in applications.

   By defining the aggregate state vector $x := \left[x_1^{\top} \dots x_N^{\top} \right]^{\top}\in\myreal^{nN}$ 
   and input vector $u := \left[u_1^{\top} \dots u_N^{\top} \right]^{\top}\in\myreal^{mN}$,
   we can write {the} interconnection \eqref{eq:interconnection} as
   \begin{subequations}
      \begin{align}
         \dot x &= (I_N \otimes A) x + (I_N \otimes B) u, \\ \label{eq:feedback}
         u      &= - (L \otimes K) x.
      \end{align}
   \end{subequations}
   The {overall} closed-loop expression is
   \begin{equation}\label{eq:closedloop}
      \dot x = \Bigl( \left(I_N \otimes A\right) - (L \otimes BK) \Bigr) x.
   \end{equation}

   {Our} goal {is to} synthesize a common static control law that enforces synchronization among systems \eqref{eq:sys}. 
   To this aim, we introduce the synchronization set:
   \begin{equation}\label{eq:setA}
      \A: = \begin{Bmatrix} x  : \, x_i - x_j = 0, \, 
      \forall i,j \in  \left \{ 1, \dots , N \right \} \end{Bmatrix}.
   \end{equation}
   We recall the definition of ``$\mu$--synchronization'' from \cite{interconnected2023}.
   \begin{definition}[$\mu$--Synchronization]{
      The attractor $\A$ in \eqref{eq:setA} is $\mu$--UGES (uniformly globally exponentially stable with rate $\mu>0$) for {system} \eqref{eq:interconnection} if there exists $M>0$ such that $|x(t)|_\A \leq M\e^{-\mu t}|x(0)|_\A$ for all $t\geq0$.
      }
   \end{definition}

   {Some of the necessary and sufficient conditions for $\mu$--synchronization in \cite[Theorem~1]{interconnected2023} are here adapted {to deal with a} synthesis {problem}: matrix} 
   $C$ {in \cite{interconnected2023}} is {replaced} by {the} {state-feedback matrix} $K$
   {in} the closed-loop dynamics \eqref{eq:closedloop}.
   Moreover, we can exploit parameter $\mu$ in iterative approaches for optimization-based selections of $K$.
   \begin{proposition}\label{th:synch}
      {Consider the system in \eqref{eq:interconnection} and the attractor 
      $\A$ in (\ref{eq:setA}).
      The synchronization set $\A$ is $\mu$--UGES {if and only if any of the following conditions holds}:}
      \begin{subequations}
         \begin{enumerate}
            \item \label{it:Ak} 
            \textnormal{[\emph{Complex condition}]}
            The complex-valued matrices  
            \begin{equation} \label{eq:Akdef}
            A_k := A - \lambda_k BK,  \quad  k=1, \dots, \nu,
            \end{equation}
            {have spectral abscissa smaller than $-\mu$}.
            \smallskip
            \item \label{it:Akreal} 
            \textnormal{[\emph{Real condition}]}
            The real-valued matrices  
            \begin{equation} \label{eq:Aekdef}
            A_{e,k} := \begin{bmatrix} A - \Re(\lambda_k) BK  & \Im(\lambda_k) BK \\
            -\Im(\lambda_k) BK & A - \Re(\lambda_k) BK 
            \end{bmatrix},
            \end{equation}
            $k=1, \dots, \nu$, {have spectral abscissa smaller than $-\mu$}.
            \smallskip
            \item \label{it:lyapReal} 
            \textnormal{[\emph{Lyapunov inequality}]}
            {For each $k=1, \dots, \nu$,}
            there exist real-valued matrices $P_k = P_k^\top \Gst0$ and $\Pi_k^\top = -\Pi_k$
            such that {matrix} $P_{e,k}  := \smat{P_k  & - \Pi_k \\ \Pi_k & P_k}\Gst0$
           {satisfies}
            \begin{align} \label{eq:Lyap_e}
               \He{P_{e,k} A_{e,k}}  \Lst {-2\mu P_{e,k}}.
            \end{align}
         \end{enumerate}
      \end{subequations}
   \end{proposition}

\section{Feedback Design}
   {We aim at designing a common state-feedback matrix  $K$ so as to ensure}
   synchronization to $\A$, {i.e., so as to satisfy} the conditions in Proposition~\ref{th:synch}.
   {We choose} an LMI{-based approach to design} $K$, {which} allows us to easily {embed additional} linear constraints in the design process.
   {Relevant linear constraints may be related, e.g., to the}
   $H_\infty$ gain, saturation {handling}, gain norm, and convergence rate \cite{BoydBookLMI1994}.

   \subsection{Revisited synchronization conditions}
      We can distinguish two main cases: either the Laplacian eigenvalues are all real or {at least one of them is} complex.
      In the {former case},
      conditions \eqref{eq:Akdef} can be framed within an LMI formulation, 
      through a procedure similar to the one we describe next.
      In the latter case, we refer to expression \eqref{eq:Lyap_e}, where the problem is lifted to a higher space, {considering $A_{e,k}$ instead of $A_k$,} so as to {work with real-}valued matrices.
      We focus on the latter case, which is more general and includes the former. 
      Let us define the inverse of $P_{e,k}$ {in \eqref{eq:Lyap_e}} 
      $$ Q_{e,k} := P_{e,k}^{-1} = Q_{e,k}^\top = \smat{Q_k & \Sigma_k \\ -\Sigma_k & Q_k},\quad k=1,\dots,\nu,$$ 
      with $Q_k$ symmetric positive definite and $\Sigma_k$ skew-symmetric. 
      {In fact, } $Q_k^{-1}=P_k-\Pi_k^\top P_k^{-1} \Pi_k$ 
      ({which is invertible, since applying} the Schur complement to $P_{e,k}$ {yields} $Q_k^{-1}\Gst0$)
      and $\Sigma_k=P_k^{-1}\Pi_k Q_k = Q_k\Pi_k P_k^{-1}$
      ({where the equality holds because}
         $
         \Pi_k P_k^{-1}Q_k^{-1} = \Pi_k - \Pi_k P_k^{-1} \Pi_k^\top P_k^{-1}\Pi_k 
         =  (P_k - \Pi_k P_k^{-1} \Pi_k^\top)P_k^{-1} \Pi_k 
         =  Q_k^{-1}P_k^{-1} \Pi_k
         $  ).
      Then, we can left- and right- multiply inequality \eqref{eq:Lyap_e} by $Q_{e,k}$, obtaining
      \begin{align} \label{eq:Lyap_Qe}
         \He{A_{e,k}Q_{e,k}}  \Lst -2\mu Q_{e,k}.
      \end{align}

      {To look for a common state-feedback matrix $K$, even when the matrices $Q_{e,k}$ are different,}
      we take advantage of the results in \cite{pipeleers2009}.
      We can rewrite \eqref{eq:Lyap_Qe} as
      \begin{equation}\label{eq:Lyap_pipeleers}
         \bmat{I_{2n} & A_{e,k}^\top} (\Phi_{\mu}\otimes Q_{e,k}) \bmat{I_{2n} \\ A_{e,k}} \Lst 0,
      \end{equation}
      where $\Phi_\mu=\smat{2\mu & 1 \\ 1 & 0}$ describes the stability region, {which} in our case is the complex {half-}plane with real part smaller than $-\mu$.
      Then, according to \cite[Section~3.1]{pipeleers2009}, \eqref{eq:Lyap_pipeleers} is equivalent to {the existence of matrices $X_{1,k},X_{2,k}\in\myreal^{2n\times 2n}$ satisfying}
      \begin{align}\label{eq:LMI_pipeleers}
         &(\Phi_\mu\otimes Q_{e,k}) +\He{ \bmat{A_{e,k}\\ -I_{2n} } \bmat{{X_{1,k}} & {X_{2,k}}} } \Lst 0,
      \end{align}
      where $X_{1,k}$ and $X_{2,k}$ are multipliers that add degrees of freedom by {de}coupling matrices $Q_{e,k}$ and $A_{e,k}$.
      {Conditions \eqref{eq:LMI_pipeleers}} {are} {still} \emph{necessary and sufficient} for $\mu$--synchronization,
      because they are equivalent to \eqref{eq:Lyap_pipeleers}.

      According to the derivation in \cite[Section~3.3]{pipeleers2009}, imposing ${X_{2,k}}=\alpha{X_{1,k}}$, with $\alpha>0$, does not add conservatism as far as {there are} $\nu$ independent matrices ${X_{1,k}}$.
      We are {now} going to relax {this condition} {by assuming that matrices $X_{1,k}$ and $X_{2,k}$ have the {specific} structure}
      $$ X_{1,k}:={X_e}Z_k, \qquad X_{2,k}:={X_e}W_k,$$ 
      where $X_e:=\smat{X&0\\0&X}=I_2\otimes X$, with $X\in\myreal^{n\times n}$, is a block-diagonal matrix common to all $\nu$ inequalities,
      {while} $Z_k,W_k\in\myreal^{2n\times 2n}$ are different multipliers for every inequality. 
      This assumption {introduces} conservativeness; therefore, the conditions are now only sufficient.
      {However,} we can now expand \eqref{eq:LMI_pipeleers} {and} obtain {the} bilinear formulation
      \begin{equation}\label{eq:LMI_Xe}
         \bmat{2\mu{Q_{e,k}} & {Q_{e,k}} \\ {Q_{e,k}} & 0} + \He{ \bmat{
            {\Theta_{k}} {Z_k }&  {\Theta_{k}} {W_k}  \\ 
                  -{X_e}{Z_k} &  -{X_e}{W_k} }  }  \Lst 0,
      \end{equation}
      with ${\Theta_{k}}=\bigl(I_2\otimes AX\bigr)- \bigl(\Lambda_k\otimes B{Y}\bigr)$, where $\Lambda_k=\smat{\alpha_k & -\beta_k \\ \beta_k & \alpha_k}$ is related to the $k$-th eigenvalue $\lambda_k=\alpha_k+\jmath\beta_k$ and $Y:=KX$ is a {suitable} change of variables.
      An expanded version of \eqref{eq:LMI_Xe} with the variables highlighted is provided in 
      equation \eqref{eq:LMI_full_tab} at the top of the next page.

      Constraints \eqref{eq:LMI_Xe} alone, might lead to badly conditioned optimized selections of $K$, due to the fact that the joint spectral abscissa of $A_{e,k}$ for all $k=1,\ldots, \nu$ may potentially grow unbounded for arbitrarily large values of $K$. Thus, as a possible sample formulation of a multi-objective optimization, we fix a maximum desired norm $\bar{\kappa}$ for $K$ and enforce the constraint $\|K\|\leq{\bar\kappa}$ through the following LMI formulation:
      \begin{equation}\label{eq:LMI_kappa}
         \bmat{X+X^\top-I & Y^\top \\ Y & {{\bar\kappa}}^2 I}\Gst0,
      \end{equation}
      stemming from the expression $K^\top K \Leq{{\bar\kappa}}^2 I$ after applying a Schur complement and exploiting $(X-I)(X^\top-I)\Geq0$.
    
      We then suggest to synthesize a state-feedback matrix $K$ satisfying $\|K\|\leq{\bar\kappa}$ and maximizing
      $\mu$ by
      {solving} the {bilinear} optimization problem
      \begin{equation}\label{eq:opt}
         \begin{aligned}
            \mathop{\max}_{\substack{
               X,Y\\
               Z_1,\dots,Z_\nu\\
               W_1,\dots,W_\nu\\
               Q_{e,1},\dots,Q_{e,\nu}}}
               \quad  \mu, \quad &\text{subject to:}\\[-25pt]
            \quad& \mbox{\eqref{eq:LMI_kappa}},\quad Q_{e,k}\Gst0, \\[5pt]
            & \mbox{BMI } \eqref{eq:LMI_Xe}, \mbox{ for all } k = 1,\dots,\nu,
         \end{aligned}
      \end{equation}
      where alternative performance-oriented linear constraints can be straightforwardly included, and then selecting $K=YX^{-1}$.
      An iterative approach can be used to make the problem quasi-convex and solve it iteratively with standard LMI techniques. 

      \begin{table*}
         \centering
         \begin{equation}\label{eq:LMI_full_tab}
            \bmat{ 2\param{\mu}\bmat{\param{Q_k} & -\param{\Sigma_k} \\ \param{\Sigma_k} & \param{Q_k}} & 
                            \bmat{\param{Q_k} & -\param{\Sigma_k} \\ \param{\Sigma_k} & \param{Q_k}}   \\ 
            \phantom{2\mu}  \bmat{\param{Q_k} & -\param{\Sigma_k} \\ \param{\Sigma_k} & \param{Q_k}} & 0} 
            +\He{  \bmat{ \bigl( (I_2\otimes A\paramA{X}) - \left(\Lambda_k\otimes B\paramA{Y}\right) \bigr) \paramB{Z_k}\phantom{\bmat{~\\~}} & 
                          \bigl( (I_2\otimes A\paramA{X}) - \left(\Lambda_k\otimes B\paramA{Y}\right) \bigr) \paramB{W_k}  \\ 
               -\bmat{\paramA{X}& 0 \\ 0 &\paramA{X}}\paramB{Z_k} & - \bmat{\paramA{X}&0\\ 0&\paramA{X}}\paramB{W_k} } } \Lst 0.
         \end{equation}
      \end{table*}

      \RestyleAlgo{ruled}
      \SetKwComment{Comment}{/* }{ */}
      \begin{algorithm}[tb]
         \caption{Iterative design of control matrix $K$}\label{alg:X_Zk}
         \KwData{$A_{e,k}, k=1, \ldots, \nu$, tolerance}
         \KwResult{$K$}
         $Z_k \gets I_{2n}$\;
         $W_k \gets \alpha I_{2n}$\;
         \While{$|\mu_{prev}-\mu_S|>\mbox{tolerance}$}{
            $\mu_{prev}=\mu_S$\;
            \texttt{-- synthesis step in Sec.~\ref{sec:iterative} --}\\
            $(\param{\mu_S},\paramA{X},\paramA{Y},\param{Q_{e,k}})=$ solve \eqref{eq:opt} given $(Z_k,W_k)$\;
            \texttt{-- analysis step in Sec.~\ref{sec:iterative} --}\\
            $(\param{\mu_A},\paramB{Z_k},\paramB{W_k},\param{Q_{e,k}})=$ solve \eqref{eq:opt} given $(X,Y)$\;

         }
         $K=YX^{-1}$\;
      \end{algorithm}

      \begin{remark}
         The most natural way to include the coefficient $\mu$ in the equations is that inspired by the techniques in \cite{pipeleers2009}, leading to the formulation in \eqref{eq:LMI_Xe}, where $\mu$ defines the stability region in the complex plane and the problem results in a generalized eigenvalue problem (GEVP).
         As an alternative, $\mu$ could be introduced as a destabilizing effect acting on the matrices $A_{e,k}$ (shifting their eigenvalues to the right in the complex plane), which are still required to be Hurwitz: 
         \begin{equation*}\label{eq:mu_inside}
            \bmat{0 & {Q_{e,k}} \\ {Q_{e,k}} & 0} +\He{ \bmat{A_{e,k}+2\mu I_{2n} \\ -I_{2n} } \bmat{{X_{1,k}} & {X_{2,k}}} } \Lst 0.
         \end{equation*}
         However, with this formulation, the problem is no longer a GEVP. 
         This complicates the implementation, since the feasibility domain with respect to $\mu$ could be bounded (while in a GEVP it is right or left unbounded)
         {and bisection is not appropriate.}
         We tested this alternative approach in simulation and we obtained similar results to those achieved with Algorithm~\ref{alg:X_Zk}, with the advantage of {typically} reaching convergence after a \emph{significantly lower} number of iterations.
      \end{remark}

   \subsection{Iterative algorithm}
\label{sec:iterative}

      In order to solve the BMI optimization problem \eqref{eq:opt} with convex techniques,
      we focus our attention on BMI \eqref{eq:LMI_Xe}, {since} the other constraints are linear, {and we propose}
      an iterative approach {for the problem solution}, described in Algorithm~\ref{alg:X_Zk}.
      The algorithm is composed of two steps:\\
      1)
      \textit{\textbf{Synthesis step}}: for {given fixed} multipliers $Z_k$ and $W_k$, $k=1,\ldots, \nu$, optimization \eqref{eq:opt} is solved in the decision variables $(\mu, X, Y,Q_{e,k})$, which corresponds to a generalized eigenvalue problem (easily solvable by a bisection algorithm) due to the fact that matrices $Q_{e,k}$ are all positive definite; \\
      2) 
      \textit{\textbf{Analysis step}}: for {given fixed} matrices $X$ and $Y$,
      optimization \eqref{eq:opt} is solved in the decision variables $(\mu, Z_k, W_k,Q_{e,k})$, which corresponds again to a generalized eigenvalue problem due to positive definiteness of $Q_{e,k}$.

      Algorithm~\ref{alg:X_Zk} essentially comprises iterations of the two steps above, until parameter $\mu$ increases less than a specified tolerance over two steps. 
      To the end of establishing a useful means of comparison in the simulations reported in Section~\ref{sec:simulations}, 
      we naively initialize the algorithm by fixing the initial multipliers as scaled identity matrices (with $\alpha>0$ properly tuned).
      More generally, we emphasize that using the Riccati construction of \cite{IsidoriBook2017_ch5}, stabilizability of $(A,B)$ is sufficient for ensuring the existence of a Riccati-based solution of \eqref{eq:opt} and it is immediate to design an infinite gain margin solution where all the matrices $A_{e,k}$ share a common quadratic Lyapunov function. 
      We do not pursue this initialization here, so as to perform a fair comparison between our algorithm 
      (initialized in a somewhat naive way) 
      and the construction resulting from \cite{IsidoriBook2017_ch5}.

      The following proposition establishes useful properties of the algorithm.

      \begin{proposition}\label{prop:alg}
         For any selection of the tolerance, if the initial condition is feasible, then
         all the iterations of Algorithm~\ref{alg:X_Zk} are feasible. Moreover, 
         $\mu$ never decreases across two successive iterations. Finally, the algorithm terminates in a finite number of steps.
      \end{proposition}

      \begin{proof}
      About recursive feasibility, note that once the first step is feasible, for any pair of successive steps, the optimal solution to the previous step is structurally a feasible solution to the subsequent step. Indeed, the variables that are frozen at each iteration are selected by using their optimal values from the previous step. Since the cost maximized at each step is always $\mu$, then $\mu$ can never decrease across subsequent steps and then recursive feasibility is guaranteed.

      About the algorithm terminating in a finite number of steps, note that the optimal value of $\mu$ is necessarily upper bounded by a maximum achievable $\overline \mu$ depending on the 
      norm of matrices $A$, $B$, on the eigenvalue of $L$ having minimum norm, and on the bound $\bar \kappa$ imposed on the norm of the state-feedback matrix $K$. 
      Since $\mu$ monotonically increases across iterations and it is upper bounded, then it must converge to a value $\mu^\star$ and eventually reach any tolerance limit across pairs of consecutive iterations.
      \end{proof}

      \begin{remark}
         Computationally speaking, each {iteration} of Algorithm~\ref{alg:X_Zk} {amounts to solving} a GEVP {because} $\mu$ is multiplying a sign definite matrix $Q_{e,k}\Gst0$, 
         {and hence} the conditions are monotonic with respect to $\mu$.
         {Therefore, we can} find the optimal $\mu$ {via} a bisection algorithm: 
         if the problem is feasible for $\mu=\mu^\star$, then the problem is feasible {for} all $\mu\leq\mu^\star$; 
         on the other hand, if the problem is infeasible {for} $\mu=\mu^\star$, then the problem is infeasible for all $\mu\geq\mu^\star$. Our objective is to find the maximum $\mu$ 
         for which the problem is feasible ({so that} no other {larger} $\mu$ leads to feasibility).
      \end{remark}

\section{Comparison and Simulations}
\label{sec:simulations}

   To test the effectiveness of Algorithm~\ref{alg:X_Zk}, we compare it with other {design} procedures that solve the simultaneous stabilization problem.
   The benchmark problem is the maximization of the rate $\mu$
   with the norm of $K$ upper bounded by ${\bar\kappa}$, as induced by constraint \eqref{eq:LMI_kappa}.

   \begin{figure}[bt]
      \centering
      \includegraphics[width=0.9\columnwidth]{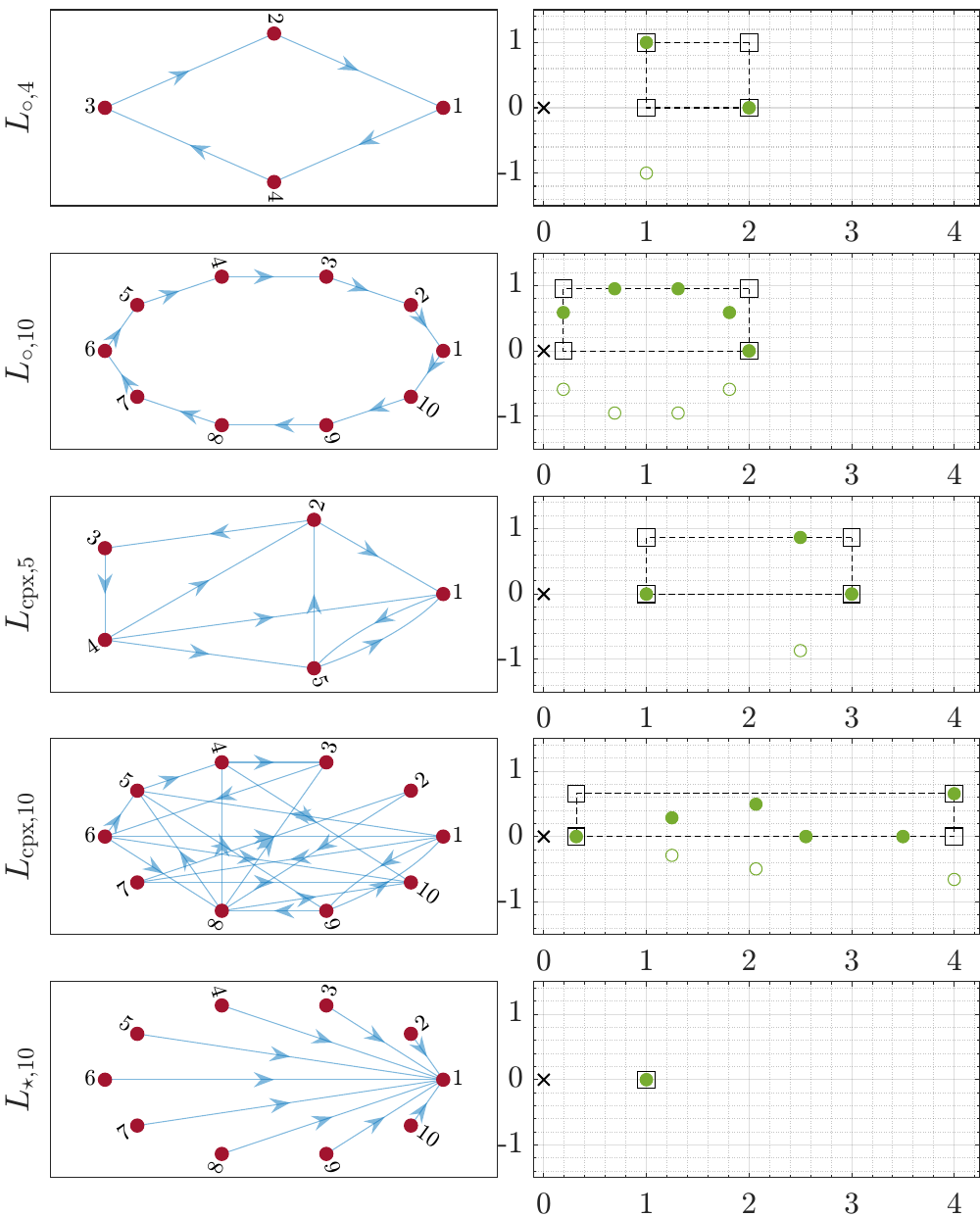}
      \caption{Left: topologies of the considered graphs. 
      Right: eigenvalues of the Laplacian matrix; the black cross denotes $\lambda_0=0$, green full dots denote all the other eigenvalues.
      The values considered in method ``b'' are visualized as squares and the relative set is delimited by a dashed line.}
      \label{fig:graph}
   \end{figure}

   \newcommand{\tabLcpxc}{  \multirow{3}{*}{$L_{\mathrm{cpx},5}$  } 
& $\hat\mu$& 0.657  & 0.654  & 0.684  & 0.686  & 1.197  & \multirow{3}{*}{14} & 1.922  & 3.853  & 4.005  & 4.016  & 4.684  & \multirow{3}{*}{5}\\ 
& $\kappa$ & 20.00  & 16.23  & 17.95  & 18.05  & 20.00  &                      & 20.00  & 16.72  & 15.93  & 15.93  & 19.99  & \\ 
& $t$ (s)   & 0.125  & 1.562  & 1.109  & 2.703  & 78.81  &                      & 0.047  & 0.516  & 0.531  & 0.719  & 9.156  &  \\ 
}
   \newcommand{\tabLcpxd}{  \multirow{3}{*}{$L_{\mathrm{cpx},10}$  } 
& $\hat\mu$& 0.107  & 0.171  & 0.178  & 0.178  & 0.374  & \multirow{3}{*}{20} & 1.466  & 1.968  & 1.987  & 1.990  & 2.209  & \multirow{3}{*}{8}\\ 
& $\kappa$ & 20.00  & 16.37  & 16.47  & 16.58  & 20.00  &                      & 20.00  & 17.43  & 17.28  & 17.28  & 19.99  & \\ 
& $t$ (s)   & 0.047  & 0.984  & 1.109  & 3.594  & 168.9  &                      & 0.031  & 0.375  & 0.562  & 2.359  & 22.17  &  \\ 
}
   \newcommand{\tabLcircq}{  \multirow{3}{*}{$L_{\circ,4}$  } 
& $\hat\mu$& 0.577  & 0.654  & 0.654  & 0.654  & 1.096  & \multirow{3}{*}{12} & 1.972  & 3.853  & 3.853  & 3.853  & 4.254  & \multirow{3}{*}{6}\\ 
& $\kappa$ & 20.00  & 16.23  & 16.23  & 16.24  & 20.00  &                      & 20.00  & 16.72  & 16.72  & 16.73  & 20.00  & \\ 
& $t$ (s)   & 0.031  & 0.953  & 0.812  & 3.000  & 56.25  &                      & 0.000  & 0.219  & 0.359  & 0.703  & 8.062  &  \\ 
}
   \newcommand{\tabLcircd}{  \multirow{3}{*}{$L_{\circ,10}$  } 
& $\hat\mu$& 0.093  & 0.075  & 0.075  & 0.075  & 0.368  & \multirow{3}{*}{94} & 1.177  & 1.403  & 1.403  & 1.402  & 1.517  & \multirow{3}{*}{18}\\ 
& $\kappa$ & 19.97  & 16.84  & 16.85  & 16.85  & 19.95  &                      & 20.00  & 17.94  & 17.94  & 17.95  & 19.95  & \\ 
& $t$ (s)   & 0.062  & 1.266  & 1.828  & 3.859  & 614.8  &                      & 0.000  & 0.391  & 0.438  & 0.781  & 34.39  &  \\ 
}
   \newcommand{\tabLstar}{  \multirow{3}{*}{$L_{\star,10}$  } 
& $\hat\mu$& 0.657  & 0.715  & 0.715  & 0.715  & 1.201  & \multirow{3}{*}{12} & 2.401  & 4.147  & 4.147  & 4.148  & 5.818  & \multirow{3}{*}{66}\\ 
& $\kappa$ & 20.00  & 19.93  & 19.93  & 19.93  & 19.99  &                      & 20.00  & 14.00  & 14.00  & 14.00  & 20.00  & \\ 
& $t$ (s)   & 0.062  & 0.922  & 0.828  & 1.344  & 36.09  &                      & 0.000  & 0.438  & 0.188  & 0.531  & 61.02  &  \\ 
}

   \begin{table*}
      \caption{Simulation results for dynamics $A_{\text{X-29}}$ and $A_{osc}$ with different network topologies.}
      \label{tab:results}
      \centering
      \footnotesize{
      \begin{tabular}{ cc|cccccc|cccccc }
         & & \multicolumn{6}{c|}{$A_{\text{X-29}}$} & \multicolumn{6}{c}{$A_{\text{osc}}$} \\
         \midrule
                & & ``a''   & ``b''    & ``c''     & ``d''  & ``e''               &  \#iter               & ``a''     & ``b''     & ``c''     & ``d''  & ``e''               &  \#iter \\  
         graph  & & Riccati & Listmann & $A_{e,k}$ & Direct & Alg.~\ref{alg:X_Zk} &  Alg.~\ref{alg:X_Zk}  & Riccati   & Listmann  & $A_{e,k}$ & Direct & Alg.~\ref{alg:X_Zk} &  Alg.~\ref{alg:X_Zk}  \\ 
         \midrule
           \multirow{3}{*}{$L_{\circ,4}$  } 
& $\hat\mu$& 0.577  & 0.654  & 0.654  & 0.654  & 1.096  & \multirow{3}{*}{12} & 1.972  & 3.853  & 3.853  & 3.853  & 4.254  & \multirow{3}{*}{6}\\ 
& $\kappa$ & 20.00  & 16.23  & 16.23  & 16.24  & 20.00  &                      & 20.00  & 16.72  & 16.72  & 16.73  & 20.00  & \\ 
& $t$ (s)   & 0.031  & 0.953  & 0.812  & 3.000  & 56.25  &                      & 0.000  & 0.219  & 0.359  & 0.703  & 8.062  &  \\ 
  \cmidrule{1-14}
           \multirow{3}{*}{$L_{\circ,10}$  } 
& $\hat\mu$& 0.093  & 0.075  & 0.075  & 0.075  & 0.368  & \multirow{3}{*}{94} & 1.177  & 1.403  & 1.403  & 1.402  & 1.517  & \multirow{3}{*}{18}\\ 
& $\kappa$ & 19.97  & 16.84  & 16.85  & 16.85  & 19.95  &                      & 20.00  & 17.94  & 17.94  & 17.95  & 19.95  & \\ 
& $t$ (s)   & 0.062  & 1.266  & 1.828  & 3.859  & 614.8  &                      & 0.000  & 0.391  & 0.438  & 0.781  & 34.39  &  \\ 
  \cmidrule{1-14}
           \multirow{3}{*}{$L_{\mathrm{cpx},5}$  } 
& $\hat\mu$& 0.657  & 0.654  & 0.684  & 0.686  & 1.197  & \multirow{3}{*}{14} & 1.922  & 3.853  & 4.005  & 4.016  & 4.684  & \multirow{3}{*}{5}\\ 
& $\kappa$ & 20.00  & 16.23  & 17.95  & 18.05  & 20.00  &                      & 20.00  & 16.72  & 15.93  & 15.93  & 19.99  & \\ 
& $t$ (s)   & 0.125  & 1.562  & 1.109  & 2.703  & 78.81  &                      & 0.047  & 0.516  & 0.531  & 0.719  & 9.156  &  \\ 
   \cmidrule{1-14}
           \multirow{3}{*}{$L_{\mathrm{cpx},10}$  } 
& $\hat\mu$& 0.107  & 0.171  & 0.178  & 0.178  & 0.374  & \multirow{3}{*}{20} & 1.466  & 1.968  & 1.987  & 1.990  & 2.209  & \multirow{3}{*}{8}\\ 
& $\kappa$ & 20.00  & 16.37  & 16.47  & 16.58  & 20.00  &                      & 20.00  & 17.43  & 17.28  & 17.28  & 19.99  & \\ 
& $t$ (s)   & 0.047  & 0.984  & 1.109  & 3.594  & 168.9  &                      & 0.031  & 0.375  & 0.562  & 2.359  & 22.17  &  \\ 
   \cmidrule{1-14}
           \multirow{3}{*}{$L_{\star,10}$  } 
& $\hat\mu$& 0.657  & 0.715  & 0.715  & 0.715  & 1.201  & \multirow{3}{*}{12} & 2.401  & 4.147  & 4.147  & 4.148  & 5.818  & \multirow{3}{*}{66}\\ 
& $\kappa$ & 20.00  & 19.93  & 19.93  & 19.93  & 19.99  &                      & 20.00  & 14.00  & 14.00  & 14.00  & 20.00  & \\ 
& $t$ (s)   & 0.062  & 0.922  & 0.828  & 1.344  & 36.09  &                      & 0.000  & 0.438  & 0.188  & 0.531  & 61.02  &  \\

         \bottomrule
        \end{tabular}
      }
   \end{table*}

   \subsection{Dynamical system and network}
      In our simulations we consider two types of agent dynamics:
      {one is} oscillatory,
      \begin{equation*}
         (A_{\text{osc}},B_{\text{osc}}) = \left( \bmat{0&-1\\1&0} , \bmat{0\\1} \right);
      \end{equation*}
      {while the second one is}
      the unstable lateral dynamics of a forward-swept wing, the Grumman X-29A, as in \cite{listmann2015},
      \begin{multline*}
         (A_{\text{X-29}},B_{\text{X-29}}) = \\  \left( 
            \smat{
               -2.059  &  0.997   & -16.55  & 0       \\
               -0.1023 & -0.0679  &  6.779  & 0       \\
               -0.0603 & -0.9928  & -0.1645 & 0.04413 \\
                1      &  0.07168 &  0      & 0.0       }
         ,  \smat{
               1.347      &  0.2365     \\
               0.09194    & -0.07056    \\
               -0.0006141 &   0.0006866 \\
               0          &  0          } 
               \right)
      \end{multline*}

      We consider five communication networks: 
      two circular graphs with $N=4$ and $N=10$ agents, 
      two generic directed graphs with $N=5$ and $N=10$ agents characterized by complex eigenvalues, 
      and a star graph with $N=10$ agents.
      The graph topologies and the eigenvalues of the associated Laplacian matrices are visualized in Fig.~\ref{fig:graph}.

   \subsection{Compared approaches}
      The approaches that we compare are {the following ones}:
      \begin{enumerate}[a.]
         \item {}{[\emph{Riccati}]}, the dual case of \cite[Section~5.5]{IsidoriBook2017_ch5}:
         the gain {is structured as} $K=B^\top P$, where $P$ is the unique solution to the algebraic Riccati equation
         \begin{equation}\label{eq:Riccati}
            A^\top P+P A-2b PBB^\top P+aI=0,
         \end{equation}
         with $b\leq\Re(\lambda_k)$ and $a>0$. 
         We solve \eqref{eq:Riccati} by fixing $b=\min(\Re(\lambda_k))$ and adjusting the value of $a$ so as to respect the bound $\|K\|\leq{\bar\kappa}$, which is easily done due to the monotonicity of $\|K\|$ with respect to $a$.
         \item {}{[\emph{Listmann}]} from \cite{listmann2015}: LMI conditions \eqref{eq:Lyap_Qe} with $Q_{e,k}=\smat{Q&0\\0&Q}$ and $Y=KQ$ are imposed for the $\lambda_k$ corresponding to the corners of a rectangular box in the complex plane that includes all non-zero eigenvalues of $L$ (considering the eigenvalues in the first quadrant is enough, since conjugate eigenvalues lead to the same condition), while incorporating in the LMI-based design the maximum norm condition \eqref{eq:LMI_kappa}. A fixed number of LMIs need to be solved regardless {of} the network size.
         \item {}{[$A_{e,k}$]}: the method resembles that in ``b'', but now conditions \eqref{eq:Lyap_Qe} are imposed for {each} $\lambda_k$, $k=1,\dots,\nu$.
         \item {}{[\emph{Direct}]}: one iteration of Algorithm~\ref{alg:X_Zk} is executed, which amounts to {solving} \eqref{eq:opt} with $Z_k=I_{2n}$, $W_k=\alpha I_{2n}$ and $\alpha>0$ properly tuned.
         Notably, matrices $Q_{e,k}$ do not have a pre-defined structure.
         \item {}{[\emph{Algorithm~\ref{alg:X_Zk}}]}: the procedure in Algorithm~\ref{alg:X_Zk} is executed up to its termination, as guaranteed by Proposition~\ref{prop:alg}.
      \end{enumerate}
      
      In the simulations, the convergence rate {of the solutions} is estimated from the spectral abscissa of matrices $A_k$ {in \eqref{eq:Akdef}},
      namely, from the {largest-real-part} eigenvalue:
      $$ \hat\mu=-\max\left(\Re\bigl(\text{eig}\left( A_k \right)\bigr)\right). $$

   \subsection{Results}
      We implement the different procedures in MATLAB, {using} the toolbox YALMIP \cite{YALMIP} {with MOSEK as an LMI solver}.
      For the algorithm, we consider a tolerance of $10^{-3}$
      and ${\bar\kappa}=20$ as the bound on the norm of $K$.

      For the different combinations of dynamics and graph topologies,
      Table~\ref{tab:results} reports a summary of all our results, along with
      estimated converge rate $\hat\mu$, norm of $K$ and execution time.
      The time evolutions of the distances from the synchronization set $\A$ are shown in Figs.~\ref{fig:Ax29} and~\ref{fig:Aosc} for the approaches ``a'', ``b'' and ``e''.

      Generally, method ``a'' {has a worse performance than} {the} considered LMI-based {methods}. 
      The gain bound is reached, but the convergence rate is the slowest, most likely because the approach is forcing
      {an infinite gain margin} for $K$. 
      Locating the eigenvalues of the $A_k$ matrices {in the complex plane} shows that {method ``a'' tends} to move to the left {a} few eigenvalues (faster modes) and penalize others, so that the convergence {speed} is limited.

      Method ``b'' performs similarly to ``c'', {as} expected, since the two methods simply consider different (eigen)values.
      {In general,} method ``b'' is more conservative than ``c'', but is faster in larger networks.
      With $L_{cpx,5}$ and $L_{cpx,10}$, method ``b'' is slightly more conservative{, as} is reasonable, since it {considers} values that are not in the spectrum of the Laplacian.

      {Method ``d''} generally {yields} better results than ``b'' and ``c'',
      {provided that} proper values of the parameter $\alpha$ {are chosen}.
      This improvement is due to the decoupling between matrices $A_{e,k}$ and $Q_{e,k}$.
      
      The best results are obtained using {our proposed} Algorithm~\ref{alg:X_Zk},
      {which gives the highest} convergence rate and {gets close to} the system specifications.
      {However}, the computational load is higher, since several iterations are needed.
      With dynamics $A_{\text{X-29}}$, the states are always converging faster; with dynamics $A_{\text{osc}}$, the performance is similar to that obtained with other non-iterative LMI-based techniques.
      Algorithm~\ref{alg:X_Zk} outperforms the other procedures in the case with dynamics $A_{\text{X-29}}$ and graph $L_{\circ,10}$: even though it requires quite some iterations to converge, it provides a controller that leads to almost one-order-of-magnitude faster convergence than the others, as shown in Fig.~\ref{fig:Ax29}.

      \begin{figure}[t]
         \centering
         \includegraphics[width=0.995\columnwidth]{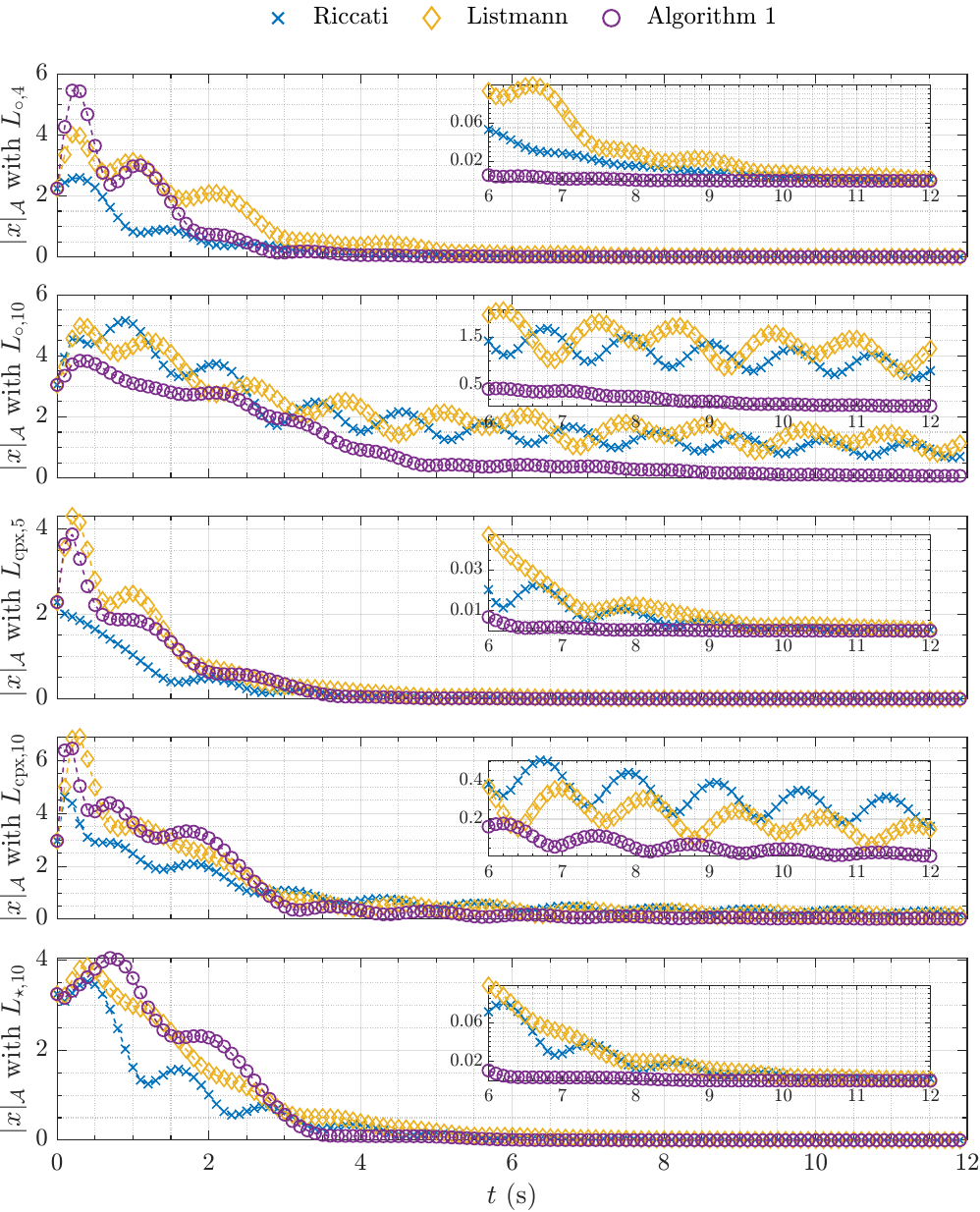}
         \caption{Time evolution of the distance of the states from the synchronization set $\A$ for agents with dynamics $A_{\text{X-29}}$. The methods ``a'', ``b'' and ``e'' are compared. The inset figures zoom into the second-half time. }
         \label{fig:Ax29}
      \end{figure}
      
      \begin{figure}[t]
         \centering
         \includegraphics[width=0.98\columnwidth]{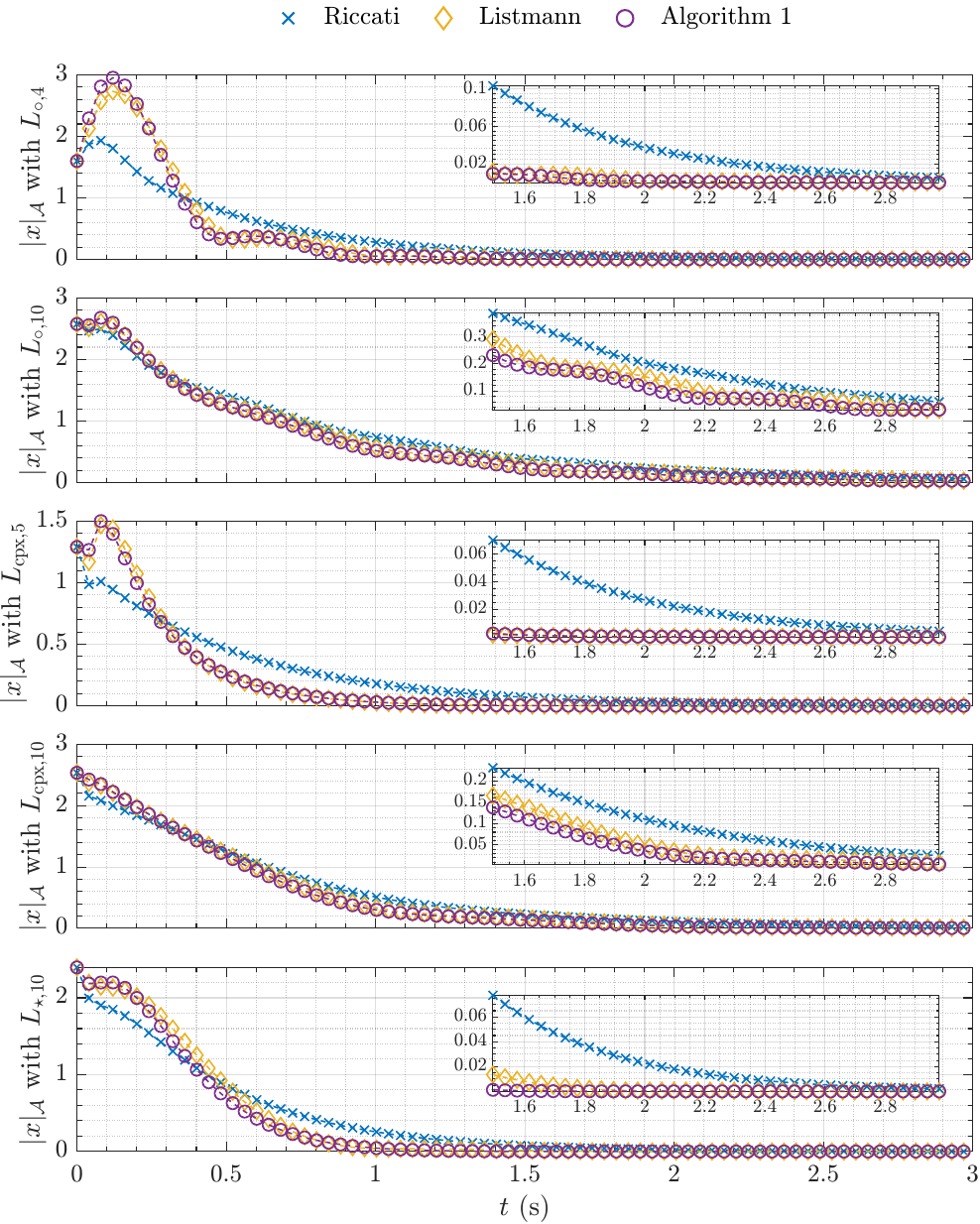}
         \caption{Time evolution of the distance  of the states from the synchronization set $\A$ for agents with dynamics $A_{\text{osc}}$. The methods ``a'', ``b'' and ``e'' are compared. The inset figures zoom into the second-half time. }
         \label{fig:Aosc}
      \end{figure}

\section{Conclusions}
   We {focused} on the synchronization of identical linear systems in the case of full-state feedback.
   First we provided some necessary and sufficient conditions for the synchronization. 
   Then, we relaxed them in order to have a new formulation that {can be iteratively solved} through LMIs. 
   This new procedure to solve the simultaneous stabilization problem, although requiring relatively large computational times, turns out to give better results in our benchmark problem where the convergence rate is maximized under given constraints on the performance.

   Our results pave the way for further developments, such as the use of alternative methods (e.g., convex-concave decomposition) to deal with BMIs 
   and the extension to the case of static output feedback control laws.

\textbf{Acknowledgment:}
   The authors would like to thank Domenico Fiore for his work in the initial stages of this research activity.

\bibliographystyle{IEEEtran}
\bibliography{2023NicolaCDC_LMI.bib}  

\begin{thebibliography}{10}
\providecommand{\url}[1]{#1}
\csname url@samestyle\endcsname
\providecommand{\newblock}{\relax}
\providecommand{\bibinfo}[2]{#2}
\providecommand{\BIBentrySTDinterwordspacing}{\spaceskip=0pt\relax}
\providecommand{\BIBentryALTinterwordstretchfactor}{4}
\providecommand{\BIBentryALTinterwordspacing}{\spaceskip=\fontdimen2\font plus
\BIBentryALTinterwordstretchfactor\fontdimen3\font minus
  \fontdimen4\font\relax}
\providecommand{\BIBforeignlanguage}[2]{{%
\expandafter\ifx\csname l@#1\endcsname\relax
\typeout{** WARNING: IEEEtran.bst: No hyphenation pattern has been}%
\typeout{** loaded for the language `#1'. Using the pattern for}%
\typeout{** the default language instead.}%
\else
\language=\csname l@#1\endcsname
\fi
#2}}
\providecommand{\BIBdecl}{\relax}
\BIBdecl

\bibitem{olfati-saber2007}
R.~Olfati-Saber, J.~A. Fax, and R.~M. Murray, ``Consensus and {{Cooperation}}
  in {{Networked Multi-Agent Systems}},'' \emph{Proceedings of the IEEE},
  vol.~95, no.~1, pp. 215--233, 2007.

\bibitem{scardovi2009}
L.~Scardovi and R.~Sepulchre, ``Synchronization in networks of identical linear
  systems,'' \emph{Automatica}, vol.~45, no.~11, pp. 2557--2562, 2009.

\bibitem{MesbahiBook2010}
M.~Mesbahi and M.~Egerstedt, \emph{Graph Theoretic Methods in Multiagent
  Networks}.\hskip 1em plus 0.5em minus 0.4em\relax {Princeton University
  Press}, 2010.

\bibitem{xiao2007}
F.~Xiao and L.~Wang, ``Consensus problems for high-dimensional multi-agent
  systems,'' \emph{IET Control Theory \& Applications}, vol.~1, no.~3, pp.
  830--837, 2007.

\bibitem{ma2010}
C.-Q. Ma and J.-F. Zhang, ``Necessary and {{Sufficient Conditions}} for
  {{Consensusability}} of {{Linear Multi-Agent Systems}},'' \emph{IEEE
  Transactions on Automatic Control}, vol.~55, no.~5, pp. 1263--1268, 2010.

\bibitem{li2010}
Z.~Li, Z.~Duan, G.~Chen, and L.~Huang, ``Consensus of {{Multiagent Systems}}
  and {{Synchronization}} of {{Complex Networks}}: {{A Unified Viewpoint}},''
  \emph{IEEE Transactions on Circuits and Systems I: Regular Papers}, vol.~57,
  no.~1, pp. 213--224, 2010.

\bibitem{dalcol2017}
L.~Dal~Col, S.~Tarbouriech, L.~Zaccarian, and M.~Kieffer, ``A consensus
  approach to {{PI}} gains tuning for quality-fair video delivery,''
  \emph{International Journal of Robust and Nonlinear Control}, vol.~27, no.~9,
  pp. 1547--1565, 2017.

\bibitem{interconnected2023}
\BIBentryALTinterwordspacing
G.~Giordano, I.~Queinnec, S.~Tarbouriech, L.~Zaccarian, and N.~Zaupa,
  ``Equivalent conditions for the synchronization of identical linear systems
  over arbitrary interconnections,'' 2023. [Online]. Available:
  \url{https://arxiv.org/abs/2303.17321}
\BIBentrySTDinterwordspacing

\bibitem{qin2015}
J.~Qin, H.~Gao, and W.~X. Zheng, ``Exponential {{Synchronization}} of {{Complex
  Networks}} of {{Linear Systems}} and {{Nonlinear Oscillators}}: {{A Unified
  Analysis}},'' \emph{IEEE Transactions on Neural Networks and Learning
  Systems}, vol.~26, no.~3, pp. 510--521, 2015.

\bibitem{IsidoriBook2017_ch5}
A.~Isidori, \emph{Lectures in {{Feedback Design}} for {{Multivariable
  Systems}}}.\hskip 1em plus 0.5em minus 0.4em\relax {Springer International
  Publishing}, 2017, ch.~5.

\bibitem{BulloBook2022}
F.~Bullo, \emph{Lectures on Network Systems}.\hskip 1em plus 0.5em minus
  0.4em\relax {Kindle Publishing}, 2022, ch.~8.

\bibitem{SaberiBook2022}
A.~Saberi, A.~A. Stoorvogel, M.~Zhang, and P.~Sannuti, \emph{Synchronization of
  {{Multi-Agent Systems}} in the {{Presence}} of {{Disturbances}} and
  {{Delays}}}.\hskip 1em plus 0.5em minus 0.4em\relax {Springer International
  Publishing}, 2022.

\bibitem{trentelman2013}
H.~L. Trentelman, K.~Takaba, and N.~Monshizadeh, ``Robust {{Synchronization}}
  of {{Uncertain Linear Multi-Agent Systems}},'' \emph{IEEE Transactions on
  Automatic Control}, vol.~58, no.~6, pp. 1511--1523, 2013.

\bibitem{listmann2015}
K.~D. Listmann, ``Novel conditions for the synchronization of linear systems,''
  in \emph{54th {{IEEE Conference}} on {{Decision}} and {{Control}}}, 2015.

\bibitem{BoydBookLMI1994}
S.~Boyd, L.~El~Ghaoui, E.~Feron, and V.~Balakrishnan, \emph{Linear {{Matrix
  Inequalities}} in {{System}} and {{Control Theory}}}.\hskip 1em plus 0.5em
  minus 0.4em\relax SIAM, 1994.

\bibitem{pipeleers2009}
G.~Pipeleers, B.~Demeulenaere, J.~Swevers, and L.~Vandenberghe, ``Extended
  {{LMI}} characterizations for stability and performance of linear systems,''
  \emph{Systems \& Control Letters}, vol.~58, no.~7, pp. 510--518, 2009.

\bibitem{YALMIP}
J.~L{\"{o}}fberg, ``Yalmip: A toolbox for modeling and optimization in
  matlab,'' in \emph{Proceedings of the CACSD Conference}, 2004.

\end{thebibliography}

\end{document}